\documentclass[a4paper,USenglish,numberwithinsect]{lipics-v2021}
\hideLIPIcs
\nolinenumbers
\makeatletter
\let\@oddfoot\@empty
\makeatother
\usepackage{mathtools,amssymb}
\usepackage{tikz}
\definecolor{proofNavigationBlue}{RGB}{24,65,94}
\newcommand{\proofnavicon}[1]{%
  \tikz[baseline=-0.5ex,x=1pt,y=1pt]{%
    \fill[proofNavigationBlue] (0,0) circle[radius=4.8pt];
    \draw[white,line width=0.8pt,line cap=round,line join=round,rotate=#1]
      (0,2.3)--(0,-2.2) (-1.65,-0.55)--(0,-2.2)--(1.65,-0.55);}}
\makeatletter
\reversemarginpar
\@mparswitchfalse
\newcommand{\proofnavmargin}[2]{%
  \marginpar{\raggedleft\hyperref[#1]{\proofnavicon{#2}}\hspace*{3pt}}}
\makeatother
\newcommand{\prooflink}[1]{%
  \mbox{}\phantomsection\label{proof-statement:#1}%
  \proofnavmargin{appendix-proof:#1}{0}\ignorespaces}
\newcommand{\proofback}[1]{%
  \mbox{}\phantomsection\label{appendix-proof:#1}%
  \proofnavmargin{proof-statement:#1}{180}\ignorespaces}
\newcommand{\parhead}[1]{\par\medskip\noindent\textbf{#1}\nobreak\hspace{0.5em}\ignorespaces}
\newcommand{\N}{\mathbb N}\newcommand{\Z}{\mathbb Z}
\newcommand{\Q}{\mathbb Q}
\newcommand{\vc}[1]{\boldsymbol{#1}}
\newcommand{\Span}{\operatorname{span}}
\title{Rational polynomial solvability reduces to stratifiability of semilinear sets}
\titlerunning{Rational polynomial solvability and stratifiability}
\author{Michaël Cadilhac}{School of Computing, DePaul University, Chicago, IL, USA}{}{https://orcid.org/0000-0001-9828-9129}{}
\authorrunning{M. Cadilhac}
\Copyright{Michaël Cadilhac}
\AtBeginDocument{\hypersetup{pdfsubject={}}}
\begin{document}\maketitle
\begin{abstract}
  A semilinear set is stratifiable if it is a finite union of linear sets whose
  periods have at most two nonzero coordinates, with no two periods in one
  component connecting alternating pairs of coordinates. We give a computable
  many-one reduction from the existence of a rational zero of an integer
  polynomial to deciding stratifiability. After the classical reduction to
  quadratic systems, we reduce rational solvability of these systems to a
  problem we call \emph{finite separation by Cartesian products of rational
    linear forms}. Its input consists of rational subspaces
  $D_1,\ldots,D_k\le\Q^r\times\Q^r$. A \emph{separator} is a map
  $(\vc x,\vc y)\mapsto(u(\vc x),v(\vc y))$, where
  $u,v\colon\Q^r\to\Q$ are nonzero rational linear forms. The problem asks whether finitely many separators suffice so
  that every pair outside $\bigcup_iD_i$ receives, from some separator, a value
  attained at no pair in that union. We then reduce this problem to stratifiability.
  Ginsburg and Spanier's characterization of letter-bounded context-free
  languages \cite[Theorem~2.1]{GS66} also induces a reduction from rational
  polynomial solvability to context-freeness of letter-bounded languages
  specified by semilinear exponent sets. An OpenAI preprint claims
  undecidability of rational polynomial solvability; if correct, our result
  settles the decidability of general stratifiability, a problem open since
  1966.
\end{abstract}

\section{Introduction}
For a set $S\subseteq\N^d$ and distinct letters $a_1,\ldots,a_d$, consider
\[
 L(S)=\{a_1^{n_1}\cdots a_d^{n_d}:(n_1,\ldots,n_d)\in S\}.
\]
The coordinates record the lengths of successive letter blocks; a language contained in $a_1^*\cdots a_d^*$ is called \emph{letter-bounded}. Suppose that $S$ is semilinear: equivalently, it is definable in first-order arithmetic with addition, or Presburger arithmetic. Can we decide whether $L(S)$ is context-free?

Ginsburg and Spanier's characterization \cite[Theorem~2.1]{GS66} expresses context-freeness of $L(S)$ as a restriction on the periods in a presentation of $S$. A linear set consists of a base vector plus arbitrary nonnegative integer combinations of finitely many period vectors, and a semilinear presentation writes $S$ as a finite union of linear sets. Such a presentation is a \emph{stratification} if each period has at most two nonzero coordinates and no two periods in the same component connect alternating positions: for $i<j<k<l$, one period cannot be nonzero at $i,k$ while another is nonzero at $j,l$. A set admitting a stratification is \emph{stratifiable}. The characterization states that $L(S)$ is context-free if and only if $S$ is stratifiable. Our decision question is therefore whether a set given by an arbitrary semilinear presentation admits a stratification.

For example, the periods $(1,1,0)$ and $(0,1,1)$ stratify
$\{(n,n+m,m):n,m\in\N\}$, whose language is $\{a^n b^{n+m}c^m:n,m\in\N\}$.
In contrast, the semilinear set
\begin{equation}\label{eq:nonstratifiable-example}
 S_\times=\{(n,m,n,m):n,m\in\N\}
\end{equation}
is not stratifiable, since the language it induces is not context-free.

\begin{samepage}
\noindent
Our main contribution is the following:
\begin{theorem}[Main result]\label{thm:main}
There is an algorithm that, given a polynomial with integer coefficients, computes a finite semilinear presentation of a set $S$ such that $S$ is stratifiable if and only if the polynomial has a rational zero.
\end{theorem}
\end{samepage}

To prove Theorem~\ref{thm:main}, we first apply the classical reduction from
polynomial equations to quadratic systems \cite{Davis}. It therefore
suffices to reduce rational solvability of finite quadratic systems to
stratifiability. We do this by two reductions through an intermediate problem.
Section~\ref{sec:overview} provides a detailed overview of the reductions.
Formal treatment starts with the preliminaries in Section~\ref{sec:prelim}.
Section~\ref{sec:transfer} then proves the first reduction and
Section~\ref{sec:geometry} the second. Section~\ref{sec:main-proof} combines
them with the standard reduction from polynomial equations to quadratic systems.
Margin arrows link to appendix proofs (\proofnavicon{0}) and back
to the main text (\proofnavicon{180}).

\parhead{Conditional undecidability consequence.}
An OpenAI preprint \cite[Theorem~1.1]{H10} claims that rational polynomial solvability is undecidable. As of 9 October 2026, we are not aware of a completed independent review of this proof. If its claim is correct, Theorem~\ref{thm:main} implies undecidability of general stratifiability. The reduction itself does not depend on that claim.

\parhead{Related work.}
Ginsburg and Spanier studied languages contained in $w_1^*\cdots w_k^*$ for fixed words $w_i$ in \cite{GS64}. Their later paper \cite[Theorems~1.3 and~2.1]{GS66} proves the effective equivalence between semilinear presentations and Presburger formulas and the stratified-period characterization used here. They decide stratifiability for a single linear set \cite[Theorem~2.2]{GS66} and explicitly leave the semilinear case open in 1966 \cite[p.~296]{GS66}. The problem also appears in Ginsburg's book \cite{Ginsburg}. Ibarra and Seki \cite{IS} relate stratifiability to equivalence questions for multitape automata and establish decidable cases. Leroux, Penelle, and Sutre \cite{LPS} prove coNP-completeness of stratifiability for sets of nonnegative integer solutions of a single finite system of linear inequalities. Our input class allows arbitrary finite unions of linear sets.

\section{Overview of the two reductions}\label{sec:overview}
We now consider a finite system $f_i(x_1,\ldots,x_n)=0$ of rational polynomial
equations of degree at most two. We reduce its rational solvability to the
following separation problem, then reduce that problem to stratifiability.

\parhead{The intermediate separation problem.}  The input is a finite family of
rational subspaces $D_1,\ldots,D_k\le\Q^r\times\Q^r$, given by spanning vectors. A
\emph{separator} is a product map $u\times v$ defined by nonzero rational linear
forms $u,v\colon\Q^r\to\Q$ by \((u\times v)(\vc x,\vc y)=(u(\vc x),v(\vc y)).\) It
\emph{separates} a pair $\vc w$ from the subspaces if
\[
 (u\times v)(\vc w)\notin\bigcup_i (u\times v)(D_i).
\]
The problem, \emph{finite separation by Cartesian products of rational linear
  forms}, asks whether there is a finite list of separators such that every
rational $\vc w\notin\bigcup_iD_i$ can be separated by one of those separators.

For example, take $D=\{((t,t),(t,s)):t,s\in\Q\}$. An outside pair $((x_1,x_2),(y_1,y_2))$ has $x_1\ne y_1$ or $x_2\ne y_1$. In the first case, the separator with $u_1(\vc x)=x_1$ and $v_1(\vc y)=y_1$ distinguishes it from $D$; in the second, use $u_2(\vc x)=x_2$ and $v_2(\vc y)=y_1$. Each separator sends $D$ to the diagonal, while at least one sends the outside pair off that diagonal. Thus $u_1\times v_1,u_2\times v_2$ are a finite separating family.
It actually suffices to construct a separator for each outside pair:
Lemma~\ref{lem:finite-subcollection} selects a finite subcollection separating
all of them, using an elementary basis argument.

\subsection{First reduction: from equations to finite separation}

\parhead{Recovering a root from one designated pair.}
We illustrate the construction with $f=x^2-1$. Let $[p]$ be the coefficient
vector of $p$ in the basis $1,x,x^2$, so $[f]=(-1,0,1)$, and write
$[p,q]=([p],[q])$ and $u[p]=u([p])$. We want separation of the \emph{designated
  pair} $\vc e=[0,1]$ by a separator \(u \times v\) to force a root. A diagonal subspace \(\Delta\) will make the
forms \(u\) and \(v\) agree after normalization, a multiplication subspace \(M\) will make their
common form be an evaluation, and an equation subspace \(E\) will make that evaluation vanish on $f$:
\[
 \begin{aligned}
 \Delta&=\Span\{[1,1],[x,x],[x^2,x^2]\},\\
 M&=\Span\{[1,x],[x,x^2]\},\qquad
 E=\Span\{[1,2],[f,f]\}.
 \end{aligned}
\]
If $u\times v$ separates $\vc e$, linearity gives its image $(0,v[1])$. Thus
$v[1]\ne0$, since every subspace image contains zero; also $u[1]\ne0$, since
otherwise $[1,1] \in \Delta$ has the same image as $\vc e$. Divide the forms by these
two values so that \(u[1]=v[1]=1\). The image of $\Delta$ contains $(1,1)$ and avoids
$(0,1)$, so it is the diagonal and the normalized forms agree. Denote their
common form by $\lambda$.

Put $c=\lambda[x]$. The two generator images of $M$, namely $(1,c)$ and
$(c,\lambda[x^2])$, must be collinear. Hence $\lambda[x^2]=c^2$, making
$\lambda[p]=p(c)$ for every polynomial \(p\) of degree at most two. Thus \(\lambda\) is evaluation at \(c\). The images $(1,2)$ and $(f(c),f(c))$ of the
generators of $E$ span the plane unless $f(c)=0$, so separation forces a
root. For $f=x^2-2$, separation would require $c^2=2$ and is therefore
impossible over $\Q$.

\parhead{Assigning image slopes to separate $[x,1]$.}
Finite separation requires that every outside pair is separated. Yet even for
$f=x^2-1$, the pair $[x,1]$ lies outside $\Delta,M,E$ and cannot be separated
(Example~\ref{ex:pointed-insufficient}). We therefore modify the subspaces to
permit its separation, while retaining the root requirement at $\vc e$.

Choose nonzero rational linear forms $u,v$ on the coefficient space of $\Q[x]_{\le2}$ and write $(u\times v)[x,1]=(\rho,\sigma)$. If $(\rho,\sigma)\ne(0,0)$, we can choose a rational slope $s_L$ for each $L\in\{\Delta,M,E\}$ with $\sigma\ne s_L\rho$. Placing each modified subspace image in its chosen line would separate $[x,1]$, provided the image of $[x,1]$ stays $(\rho,\sigma)$. To do this, we give $v$ an independently adjustable coefficient for every generator.

Number the generators of each $L$ in the listed order, starting at $d=0$, and
write them as $[p_{L,d},q_{L,d}]$. Introduce a distinct polynomial variable
$b_{L,d}$ for each index $(L,d)$ and append these monomials to the coefficient
basis. Thus $[b_{L,d}]$ is a new unit vector; coefficient vectors of polynomials in $\Q[x]_{\le2}$ acquire zero coefficients at every $b_{L,d}$. Replace each generator by
$[p_{L,d},q_{L,d}+b_{L,d}]$, and call their span $\widehat{L}$. For example,
$\widehat{E} = \Span\{[1,2+b_{E,0}], [f, f+b_{E,1}]\}$.  Extend $v$ by assigning, for every $L,d$,
\[
 v[b_{L,d}]=s_Lu[p_{L,d}]-v[q_{L,d}].
\]
These are values on independent new basis vectors; the values $u[b_{L,d}]$ can still be chosen independently. Every modified generator has image $(u[p_{L,d}],s_Lu[p_{L,d}])$, so $\widehat{L}$ maps into the line of slope $s_L$. Since $[x,1]$ has zero coefficients at all added variables, its image stays $(\rho,\sigma)$ and avoids these lines. For example, choose $u[1]=1$, $u[x]=2$, $u[x^2]=0$, let $v$ extract the constant coefficient on $\Q[x]_{\le2}$, and take every $s_L=3$. The image of $[x,1]$ is then $(2,1)$, off every chosen line.

\parhead{Preserving the root requirement at $\vc e$.}
The freedom to choose image slopes must preserve the implication from
separation of $\vc e$ to a root. We therefore introduce one last variable $t$,
giving the coefficient basis $1,x,x^2,t,(b_{L,d})$, and add
\[
 H=\Span\{[t,1]\},\qquad
 K=\Span\bigl(\{[t,0]\}\cup\{[b_{L,d},b_{L,d}]:L,d\}\bigr).
\]
The pair $\vc e$ still lies outside every subspace. A separator of $\vc e$ must
have $u[t]\ne0$, since otherwise the generator of $H$ has the same image
$(0,v[1])$ as $\vc e$. The image of $K$ then contains the nonzero horizontal
vector $(u[t],0)$. Any nonzero value $v[b_{L,d}]$ would add a nonhorizontal
vector and make that image the whole plane. Thus $v[b_{L,d}]=0$ for every
$L,d$. Each modified generator has the same image as
$[p_{L,d},q_{L,d}]$, so the restrictions of $u,v$ to $\Q[x]_{\le2}$ separate
$\vc e$ from $\Delta,M,E$ and yield a root by the earlier calculation.

Conversely, at a rational root $c$, let both forms evaluate polynomials in
$\Q[x]_{\le2}$ at $c$:
\begin{equation}\label{eq:overview-root-evaluation}
 u[p]=v[p]=p(c)\qquad(p\in\Q[x]_{\le2}).
\end{equation}
Extend them by setting
\begin{equation}\label{eq:overview-root-extension}
 u[b_{L,d}]=v[b_{L,d}]=v[t]=0\quad\text{for every }L,d,
 \qquad u[t]=\tau>0.
\end{equation}
For every positive rational $\tau$, the images of $\widehat\Delta,\widehat M,
\widehat E$ have finite slopes $1,c,2$, the image of $K$ is horizontal, and
$H$ maps to $\Q(\tau,1)$. All avoid the nonzero vertical image of every
nonzero multiple of $\vc e$.

For the earlier separator of $[x,1]$, instead complete the assignment with
\[
 u[t]=v[t]=0,\qquad
 u[b_{L,d}]=v[b_{L,d}]\quad\text{for every }L,d.
\]
Now $H$ maps to the vertical axis and $K$ into the diagonal. The image
$(2,1)$ avoids both and the slope-$3$ images of the $\widehat L$, so
$u\times v$ separates $[x,1]$ from the enlarged family.

More generally, the same extension rule separates every $[p,q]\notin\Q\vc e$
with $p,q\in\Q[x]_{\le2}$, choosing the forms and slopes afresh for each pair.
If $p\ne0$, choose $v[1]=1$ and $u[p]\notin\{0,v[q]\}$, then choose every
$s_L\ne v[q]/u[p]$: the image avoids the vertical axis, the diagonal, and
the chosen slope lines. If $p=0$ and $q$ is nonconstant, choose $u[1]=1$,
$v[1]=0$, and $v[q]\ne0$, then recompute all added values by the same rule.
Now $H$ maps to zero, while the nonzero vertical image $(0,v[q])$ avoids
the diagonal containing the image of $K$ and the finite-slope images of
all $\widehat L$. Together with the root-based separators for the nonzero
multiples of $\vc e$, this separates every nonzero pair $[p,q]$ with
$p,q\in\Q[x]_{\le2}$ from the enlarged family. Among these pairs, a root is
required only for the nonzero multiples of $\vc e$.

\parhead{Separating outside pairs with added variables.}
We now seek separators for the outside pairs with a nonzero coefficient at
$t$ or some $b_{L,d}$ in at least one entry. We first construct separators
that do not require a root. Every constructed subspace has zero coefficient
of $t$ in its second entry. If an outside pair's second entry has a nonzero
coefficient of $t$, let $v$ extract that coefficient and take any nonzero
$u$. Its image is off the horizontal axis containing every subspace image,
so $u\times v$ separates the pair from all the constructed subspaces.

For pairs whose second entry has zero coefficient of $t$, changing the
slopes can change the pair's image through its $b_{L,d}$ terms. We choose
the values of $u,v$ on $1,x,x^2$ jointly with the slopes, extend $v$ by the
earlier formula, and set
$u[b_{L,d}]=v[b_{L,d}]$ and $u[t]=v[t]=0$. Setting $u[t]=0$ permits nonzero
$v[b_{L,d}]$: otherwise the corresponding generator of $K$, together with
$[t,0]$, would have images spanning the plane. With $v[1]=1$, the image of
$H$ is vertical and that of $K$ lies in the diagonal, so an image
$(\rho,\sigma)$ outside these and the chosen slope lines gives a separator.
For such a pair with first entry $[\alpha t]$ whose full second entry is not a
scalar multiple of $[1]$, choose new values on $1,x,x^2$ and new slopes,
subject to $u[1]=1$ and $v[1]=0$, and recompute all added values by the
extension formula. Some such choice gives the pair a nonzero vertical image,
while $H$ maps to zero and the other subspace images have finite slopes.

These slope constructions separate every outside pair whose second entry
has zero coefficient of $t$, except possibly pairs $[\alpha t,\beta]$
and pairs $\vc g+\alpha[t,0]$ with $\vc g\in\widehat L$.

For $[\alpha t,\beta]$, being outside $K$ and $H$ gives
$\beta\ne0$ and $\beta\ne\alpha$. If $\alpha=0$, the pair is a nonzero
multiple of $\vc e$, already separated under the root hypothesis.
If $\alpha\ne0$, let $u$ extract the coefficient of $t$ and let $v$
extract the constant coefficient. The image $(\alpha,\beta)$ avoids
the vertical images of the $\widehat L$, the horizontal image of $K$,
and the diagonal image of $H$, so these forms separate the pair.

The slope extensions cannot separate an outside pair
$\vc g+\alpha[t,0]$, with $\alpha\ne0$ and $\vc g\in\widehat L$, from
$\widehat L$: setting $u[t]=0$ sends it to the same image as $\vc g$.
Instead, fix $u$ to extract the coefficient of $t$ and let $v$ vary on the
coefficient space of $\Q[x]_{\le2}$, extending it by zero on every added
basis vector. We determine when these choices separate the shifted pair
from the whole family. Write
$\vc g=[p,q+\sum_d\theta_db_{L,d}]$, where $[p,q]\in L$.
Here $q\ne0$: the second-entry polynomials of the listed generators of $L$
are linearly independent, so $q=0$ would force all $\theta_d=0$ and
$\vc g=0$, placing the shifted pair in $K$.

The pair's image is $(\alpha,v[q])$, off the vertical images of all the
$\widehat L$.  The images of $K$ and $H$ are the horizontal axis and
$\Q(1,v[1])$, respectively. Thus $u\times v$ separates the pair exactly when
$v[q]\ne0$ and $v[q]-\alpha v[1]\ne0$.  If $q\ne\alpha$, the polynomials $q$ and
$q-\alpha$ are both nonzero, so some rational linear form satisfies both
inequalities. This separates every such shifted pair without a root.  When
$q=\alpha$, however, the image is always $\alpha(1,v[1])$, in the image of $H$; no choice
of $v$ completes a separator with this fixed $u$.  Such a pair is
\[
 [p+\alpha t,\ \alpha+\sum_d\theta_db_{L,d}].
\]
Removing the $t$-shift gives $\vc g\in\widehat L$. By the construction of
$\widehat L$, this means $[p,\alpha]\in L$, and the $\theta_d$ are the
coefficients of $[p,\alpha]$ in the listed generators of $L$. Determining
$[p,\alpha]$ therefore determines the added terms too, and hence the
whole shifted pair.

In $\Delta$, the condition $[p,\alpha]\in\Delta$ forces $p=\alpha$.
Only the generator $[1,1]$ occurs, so
$\vc g=\alpha[1,1+b_{\Delta,0}]$.
In $M$, the second entry is divisible by $x$, so it cannot equal the
nonzero constant $\alpha$.
In $E$, write $[p,\alpha]=[a+\kappa f,2a+\kappa f]$, with $a,\kappa\in\Q$.
Since $f=x^2-1$, a constant second entry forces $\kappa=0$ and
$a=\alpha/2$. Only the generator $[1,2]$ occurs, so
$\vc g=\alpha[\tfrac12,1+\tfrac12b_{E,0}]$.
Restoring the shift $\alpha[t,0]$ therefore gives the following two
shifted lines:
\[
 \Q[1+t,1+b_{\Delta,0}],\qquad
 \Q[\tfrac12+t,1+\tfrac12b_{E,0}].
\]
The preceding constructions separate every outside pair off these
two lines, using a rational root only to separate the nonzero multiples
of $\vc e$. We now choose the evaluation extension to separate the
nonzero pairs on these two lines as well.

\parhead{Choosing the root extension for the two shifted lines.}  Return to the
evaluation extension defined by \eqref{eq:overview-root-evaluation} and
\eqref{eq:overview-root-extension}, with $u[t]=\tau>0$. The two shifted
representatives have images \((1+\tau,1)\) and \((\tfrac12+\tau,1)\).  Their nonzero
second coordinates keep them off the horizontal image of $K$.  Their
first-coordinate offsets $1$ and $\tfrac12$ keep them off $\Q(\tau,1)$, the image
of $H$. Each finite-slope $\widehat L$ image excludes at most one value of
$\tau$ for each representative. Choose one positive rational $\tau$ outside these
finitely many exclusions. The resulting separator separates every nonzero
multiple of both representatives and preserves the separation of the nonzero
multiples of $\vc e$ already proved.

A rational root therefore supplies separators for every outside pair, and
Lemma~\ref{lem:finite-subcollection} supplies a finite list. Conversely, such
a list separates $\vc e$ and yields a root. Section~\ref{sec:transfer} gives
the construction for arbitrary quadratic systems.

\subsection{Second reduction: from finite separation to stratifiability}
For a family $D_1,\ldots,D_k$, represent integer pairs outside its union by differences of nonnegative coordinates. The output is
\[
 S=\left\{(\vc x^+,\vc x^-,\vc y^+,\vc y^-)\in\N^{4r}:
 (\vc x^+-\vc x^-,\vc y^+-\vc y^-)\notin\bigcup_iD_i\right\}.
\]
Membership is given by linear equations and their negations, so Lemma~\ref{lem:semilinearity} computes a semilinear presentation of $S$. The order of the four coordinate blocks will be used when checking stratifiability.

For each separator $u\times v$, collect the elements of $S$ whose represented pairs it separates. Membership in this subset depends on two weighted sums:
\[
 u(\vc x^+-\vc x^-),\qquad v(\vc y^+-\vc y^-).
\]
After clearing the coefficient denominators of $u,v$, the permitted values satisfy a Presburger condition on two integer-weighted sums over consecutive coordinate blocks. Lemma~\ref{lem:weighted} shows that the associated letter-bounded language is context-free; Ginsburg and Spanier's characterization \cite[Theorem~2.1]{GS66} makes this subset stratifiable. Thus a finite separating family for all outside pairs gives a finite stratifiable cover of $S$.

Conversely, suppose $S$ is stratifiable. We extract a finite separating family for $D_1,\ldots,D_k$ from a stratification. Increasing matching positive and negative coordinates by the same amount leaves the represented pair unchanged. The representations of each outside integer pair therefore contain a grid of these cancellation directions. In any finite stratification, some component meeting this grid has nonnegative period combinations with signed difference zero spanning all these directions.

Consider the periods occurring in these combinations. Their signed images can be added and subtracted within the set of pairs represented by the component. The restriction on alternating supports ensures that every linear form vanishing on the rational span of these images is a linear combination of one fixed form $u$ on the first entry and one fixed form $v$ on the second. The forms can be chosen nonzero, and $\ker u\times\ker v$ is therefore contained in that span. Fix such a separator for every component with the stated spanning property. Fix a rational pair outside the subspaces. If a selected separator fails to separate it from $D_i$, the pair differs from a point of $D_i$ by a rational combination of those period images. Choose one such expression for each candidate component and each $D_i$ whenever possible. Clear all these coefficients and the pair's coordinates by one common integer multiple, then select a component representing that multiple. Subtracting the resulting integer combination from the represented multiple would put a point of $D_i$ among the pairs represented by this component, a contradiction. Thus the fixed separators separate every outside rational pair (Lemma~\ref{lem:finite-necessity}). The proof adapts the minimal-element argument of Ginsburg and Spanier \cite[proof of Theorem~1.3]{GS66} and the structure of stratified periods studied by Brough \cite{Brough}.

\section{Preliminaries}\label{sec:prelim}
Put $\N=\{0,1,2,\ldots\}$.
Bold letters denote vectors and tuples; their coordinates are written in ordinary type.
A \emph{linear set} in $\N^d$ is a set
\[
 \vc{a}+P^*=\left\{\vc{a}+\sum_{i=1}^\ell n_i\vc{t}_i:n_i\in\N\right\},
 \qquad \vc{a}\in\N^d,\quad P=\{\vc{t}_1,\ldots,\vc{t}_\ell\}\subseteq\N^d.
\]
The vector $\vc{a}$ is its \emph{base}, and the elements of $P$ are its \emph{periods}. A \emph{semilinear set} is a finite union of linear sets. A \emph{semilinear presentation} lists the base and periods of each component.

A finite period set $P$ is \emph{stratified} if
\begin{enumerate}
\item every $\vc{t}\in P$ has at most two nonzero coordinates; and
\item there are no $\vc{t}_1,\vc{t}_2\in P$ and $i<j<k<l$ with $t_{1,i}t_{1,k}t_{2,j}t_{2,l}\ne0$.
\end{enumerate}
A \emph{stratification} of a set is a finite semilinear presentation whose period sets are stratified. A set is \emph{stratifiable} if it has a stratification.

Presburger arithmetic is first-order arithmetic on the integers, or on $\N$, with addition and order but without multiplication of variables. A subset of $\N^d$ is semilinear if and only if it is Presburger-definable. One can effectively convert between these formulas and semilinear presentations \cite[Theorem 1.3]{GS66}. We also use the following classical theorem \cite[Theorem 2.1]{GS66}.
\begin{theorem}[Ginsburg and Spanier, {\cite[Theorem~2.1]{GS66}}]\label{thm:GS}
For distinct letters $a_1,\ldots,a_d$ and $S\subseteq\N^d$, the \emph{letter-bounded language}
$L(S)=\{a_1^{n_1}\cdots a_d^{n_d}:\vc{n}\in S\}$ is context-free if and only if $S$ is stratifiable.
\end{theorem}

All vector spaces are over $\Q$. We write $\Span$ for linear span, $U\le V$ for a subspace, and $\dim V$ for dimension. A linear form $u\colon\Q^r\to\Q$ is a weighted sum $u(\vc z)=\sum_i c_i z_i$ with $c_i\in\Q$. For a linear map $T$, $T(E)$ and $T^{-1}(E)$ denote image and inverse image.

Write $\Q[\vc z]_{\le d}$ for polynomials of total degree at most $d$, whose monomials of degree at most $d$ are a basis. We use $[p]$ for the coefficient vector in a specified ordered monomial basis, $[p,q]=([p],[q])$, and $u[p]=u([p])$.

The source problem for our construction is the following.
\begin{quote}
\textbf{Solvability of quadratic systems over $\Q$}
\par\smallskip\noindent\textbf{Input:} A finite list $f_1,\ldots,f_m\in\Q[x_1,\ldots,x_n]$ of polynomials of total degree at most two.
\par\smallskip\noindent\textbf{Question:} Is there a tuple $\vc{c}\in\Q^n$ such that $f_i(\vc{c})=0$ for every $i$?
\end{quote}
Rational coefficients are represented by integer numerators and nonzero integer denominators.

\section{From quadratic systems to finite separation}\label{sec:transfer}
\begingroup\widowpenalty=10000\clubpenalty=10000
Recall that a \emph{separator} is the map $u\times v\colon(\vc x,\vc y)\mapsto(u(\vc x),v(\vc y))$ defined by nonzero rational forms $u,v\colon\Q^r\to\Q$. It separates $\vc w$ from a family of subspaces when $(u\times v)(\vc w)$ lies outside all their images under $u\times v$.
\par\endgroup

\begin{samepage}
\begin{quote}
\textbf{Finite separation by Cartesian products of rational linear forms}
\par\nobreak\smallskip\noindent\textbf{Input:} An integer $r\ge1$ and finitely many rational subspaces \mbox{$D_1,\ldots,D_k\le\Q^r\times\Q^r$}, each specified by finitely many rational spanning vectors.
\par\nobreak\smallskip\noindent\textbf{Question:} Is there a finite list of separators $u_1\times v_1,\ldots,u_h\times v_h$ such that, for every rational pair $\vc w$ outside the subspaces, some $j$ satisfies
\[
 (u_j\times v_j)(\vc w)\notin\bigcup_{i=1}^k (u_j\times v_j)(D_i)?
\]
\end{quote}
\end{samepage}

\begin{theorem}[First reduction]\label{thm:quadratic-transfer}
There is a computable many-one reduction from rational solvability of quadratic systems over $\Q$ to finite separation by Cartesian products of rational linear forms.
\end{theorem}

\begin{proof}
Write the input equations as $f_i=0$. First test by rational linear algebra whether $1\in\Span_{\Q}\{f_1,\ldots,f_m\}$. If so, the system has no rational solution; replace it by $x^2-2=0$, which also has no rational solution and whose equation span does not contain $1$. We may therefore assume $1\notin\Span_{\Q}\{f_1,\ldots,f_m\}$.

The construction in Sections~\ref{sec:pointed-construction} and~\ref{sec:enlarged-construction} computes from the input equations a finite family $\mathcal D$ of rational subspaces and a designated pair outside their union. Lemmas~\ref{lem:recover-pointed} and~\ref{lem:pointed-quadratic} show that any separator of this pair yields a rational solution. Thus a finite separating family for $\mathcal D$ yields a solution.

Conversely, under the assumption on the equation span, Lemma~\ref{lem:root-free-separation} separates every outside pair off three explicit rational lines, by choosing image slopes and extracting coefficients of the added variable $t$. If a rational solution exists, the evaluation extension in Lemma~\ref{lem:pointwise-completion} separates every nonzero pair on those three lines with one choice of its positive parameter. Lemma~\ref{lem:finite-subcollection} then selects finitely many separators serving all outside pairs.
\end{proof}

The following fact about a fixed finite family of rational subspaces justifies this finite selection. It allows the forms in a separation proof to depend on the pair being separated.

\begin{lemma}[A finite subcollection suffices]\label{lem:finite-subcollection}\prooflink{lem:finite-subcollection}
Fix rational subspaces $D_1,\ldots,D_k\le\Q^r\times\Q^r$. Every collection of separators has a finite subcollection separating exactly the same pairs from $D_1,\ldots,D_k$.
\end{lemma}

\subsection{Making separation at one pair equivalent to a rational root}\label{sec:pointed-construction}
Let the input equations be $f_i(\vc x)=0$, $1\le i\le m$, where $\vc x=(x_1,\ldots,x_n)$. Throughout this construction assume $1\notin\Span_{\Q}\{f_1,\ldots,f_m\}$, as permitted by the preliminary step in the proof of Theorem~\ref{thm:quadratic-transfer}. Fix an ordering of the monomials $1$, $x_j$, and $x_jx_k$ with $j\le k$ in $\Q[\vc x]_{\le2}$, beginning with $1$. Represent polynomials by their coefficient vectors in this basis, as in Section~\ref{sec:prelim}.

To recover a common root, we make the separator determine one evaluation of the input polynomials. A diagonal subspace makes its forms agree after normalization; multiplication subspaces make their common form respect products; an equation subspace makes this evaluation vanish on every $f_i$. Define
\begin{align}
 \Delta&=\{[p,p]:p\in\Q[\vc x]_{\le2}\},
                 \label{eq:pointed-diagonal}\\
 M_j&=\{[p,x_jp]:p\in\Q[\vc x]_{\le1}\}
                 &&(1\le j\le n),\label{eq:pointed-multiplication}\\
 E&=\Span\bigl(\{[1,2]\}\cup\{[f_i,f_i]:1\le i\le m\}\bigr).
                 \label{eq:pointed-equations}
\end{align}
Let $\mathcal L=\{\Delta,M_1,\ldots,M_n,E\}$. Set $\vc e=[0,1]$. A separator of this designated pair will determine a rational solution.

\begin{lemma}\label{lem:pointed-quadratic}\prooflink{lem:pointed-quadratic}
The pair $\vc e$ can be separated from $\mathcal L$ by a separator if and only if the quadratic system has a rational solution.
\end{lemma}

\begin{example}[Why the initial family must be enlarged]\label{ex:pointed-insufficient}\prooflink{ex:pointed-insufficient}
For $f=x^2-1$, evaluation at $1$ separates $\vc e$ from $\Delta,M_1,E$, but the outside pair $[x,1]$ cannot be separated.
\end{example}

\subsection{Enlarging the space while preserving the root requirement at \texorpdfstring{$\vc e$}{e}}\label{sec:enlarged-construction}

We enlarge the space so that a rational solution supplies separators for every outside pair. Giving each basis generator its own added coordinate in the second entry lets $v$ move the modified generator's image vertically, independently of the other generators. These adjustments place each modified subspace image in a chosen line through the origin of finite rational slope. A separate coordinate and two additional subspaces ensure that separating $\vc e$ forces $v$ to vanish on every coordinate added to a generator. Each modified generator then has the same image as the generator it replaces, so separation of $\vc e$ still forces a rational root.

Fix a basis $[p_{L,d},q_{L,d}]$, $0\le d<\dim L$, for each $L\in\mathcal L$. For $\Delta$, take $[\mu,\mu]$ over the monomials, with $[1,1]$ first; for $M_j$, take $[\mu,x_j\mu]$ over $\mu\in\{1,x_1,\ldots,x_n\}$. For $E$, compute a basis $[\widetilde f_1],\ldots,[\widetilde f_s]$ of $\Span\{[f_1],\ldots,[f_m]\}$ and take $[1,2],[\widetilde f_1,\widetilde f_1],\ldots,[\widetilde f_s,\widetilde f_s]$, in this order. A sum over $L,d$ ranges over all these basis vectors; in particular, $d=0$ indexes $[1,1]$ for $\Delta$ and $[1,2]$ for $E$.

The generators specified for $E$ are independent: subtracting the entries of a zero linear combination forces the coefficient of $[1,2]$ to vanish, and independence of the $[\widetilde f_i]$ forces the other coefficients to vanish.

Introduce a polynomial variable $t$ and one variable $b_{L,d}$ for every basis index $(L,d)$, and work in the vector space
\[
 \Q[\vc x]_{\le2}+\Q[t,(b_{L,d})]_{\le1}.
\]
The summands share $1$. Extend the monomial basis of $\Q[\vc x]_{\le2}$ by $t$ and all $b_{L,d}$, obtaining a basis of length $r$ for both entries. Each coefficient vector is uniquely
\[
 \left[p+\alpha t+\sum_{L,d}\eta_{L,d}b_{L,d}\right],
 \qquad p\in\Q[\vc x]_{\le2},\quad \alpha,\eta_{L,d}\in\Q.
\]
For $p\in\Q[\vc x]_{\le2}$, the vector $[p]$ has zero coefficients at $t$ and every $b_{L,d}$; the designated pair in the enlarged space is still $\vc e=[0,1]$.

Replacing $[p_{L,d},q_{L,d}]$ by $[p_{L,d},q_{L,d}+b_{L,d}]$ changes its second image value by $v[b_{L,d}]$. To make each modified generator map to $(u[p_{L,d}],v[q_{L,d}])$ whenever $u\times v$ separates $\vc e$, the line $H$ below forces $u[t]\ne0$, allowing $K$ to force every $v[b_{L,d}]=0$.

Accordingly, define the following subspaces of $\Q^r\times\Q^r$:
\begin{align}
 \widehat{L}&=\Span\{[p_{L,d},q_{L,d}+b_{L,d}]:0\le d<\dim L\}
                      &&(L\in\mathcal L),\label{eq:lifted-subspaces}\\
 H&=\Span\{[t,1]\},\label{eq:lifted-anchor}\\
 K&=\Span\bigl(\{[t,0]\}\cup\{[b_{L,d},b_{L,d}]:L,d\}\bigr).
                      \label{eq:lifted-gate}
\end{align}
Let $\mathcal D=\{\widehat{L}:L\in\mathcal L\}\cup\{H,K\}$ be the resulting family of $n+4$ subspaces.

\begin{lemma}[Recovery at the designated pair $\vc e$]\label{lem:recover-pointed}\prooflink{lem:recover-pointed}
The pair $\vc e=[0,1]$ lies outside every member of $\mathcal D$. If a separator $u\times v$ separates $\vc e$ from $\mathcal D$, the restrictions of $u$ and $v$ to the coefficient vectors of $\Q[\vc x]_{\le2}$ are nonzero and separate $\vc e$ from $\mathcal L$.
\end{lemma}

\begin{samepage}
\begin{example}[Why equal forms do not suffice]\label{ex:equal-forms}\prooflink{ex:equal-forms}
For separators of $\vc e$, normalization makes the restrictions of $u,v$ to $\Q[\vc x]_{\le2}$ agree. Equality of the full forms cannot be required for all outside pairs: $[1,1]$ lies outside every member of $\mathcal D$, yet no $u\times u$ separates it. The problem must therefore allow $u$ and $v$ to differ, even when a rational root exists.
\end{example}
\end{samepage}

The constructed instance of finite separation is $\mathcal D$. Its finite separating families yield rational solutions by Lemmas~\ref{lem:recover-pointed} and~\ref{lem:pointed-quadratic}. For the converse, put $Y=\Q^{2r}\setminus\bigcup_{D\in\mathcal D}D$. We first separate all pairs in $Y$ off three explicit lines, using only the assumption on the equation span. A rational solution supplies one separator for the nonzero pairs on those lines; Lemma~\ref{lem:finite-subcollection} then gives a finite separating family.

\subsection{Separating pairs without using a root}\label{sec:background-cover}
The slope extensions keep $u[t]=0$ to permit nonzero $v[b_{L,d}]$ while keeping the image of $K$ proper. These extensions send $\vc g+\alpha[t,0]$ to the same point as $\vc g\in \widehat{L}$. We then choose $u$ to extract the coefficient of $t$ and set every $v[b_{L,d}]=0$, which permits $u[t]=1$ while keeping the image of $K$ horizontal. The proof first treats the pairs $[\alpha t,\beta]$, then the shifted pairs $\vc g+\alpha[t,0]$.

\begin{lemma}\label{lem:root-free-separation}
Every pair in $Y$ outside the union of the three lines
\begin{equation}\label{eq:three-lines}
 \Q\vc e,\qquad
 \Q[1+t,1+b_{\Delta,0}],\qquad
 \Q[\tfrac12+t,1+\tfrac12b_{E,0}]
\end{equation}
can be separated from $\mathcal D$.
\end{lemma}
\begin{proof}
The coefficient extraction and slope extensions verified in
\prooflink{coarse-classification}Appendix~\ref{app:slope-separation}
separate all pairs in $Y$ except possibly
\begin{equation}\label{eq:coarse-residuals}
 [\alpha t,\beta],\qquad
 \vc g+\alpha[t,0],\quad
 \vc g=\left[p,q+\sum_d\theta_db_{L,d}\right]\in \widehat{L},
\end{equation}
where $\beta\ne0$, $\beta\ne\alpha$ in the first case, and $\alpha\ne0$, $q\ne0$ in the second.

For $[\alpha t,\beta]$, the case $\alpha=0$ gives a nonzero multiple of $\vc e$, on the first line in~\eqref{eq:three-lines}. If $\alpha\ne0$, let $u$ and $v$ extract the coefficients of $t$ and $1$, respectively. The images of all the $\widehat L$ lie on the vertical axis, the image of $K$ is horizontal, and the image of $H$ is diagonal. The pair's image $(\alpha,\beta)$ avoids all of them because $\alpha\ne0$, $\beta\ne0$, and $\beta\ne\alpha$. Both forms are nonzero.

Now consider a shifted pair $\vc g+\alpha[t,0]$ in~\eqref{eq:coarse-residuals}. Fix $u$ to extract the coefficient of $t$, and let $v$ vary on the coefficient space of $\Q[\vc x]_{\le2}$, extending it by zero on every added basis vector. For every $J\in\mathcal L$, the $\widehat J$ generators map to $(0,v[q_{J,d}])$; the image of $K$ is horizontal, and $H$ maps to $\Q(1,v[1])$. The shifted pair maps to $(\alpha,v[q])$ and already avoids every $\widehat J$ image since $\alpha\ne0$. Thus it is separated exactly when
\begin{equation}\label{eq:reader-conditions}
 v[q]\ne0,\qquad v[q]-\alpha v[1]\ne0.
\end{equation}
By linearity, the second inequality is $v[q-\alpha]\ne0$. If $q\ne\alpha$, both $q$ and $q-\alpha$ are nonzero. The product $v[q]v[q-\alpha]$, viewed as a polynomial in the coefficients of $v$, is therefore nonzero at some rational tuple, giving a separator. Both forms are nonzero since $u[t]=1$ and $v[q]\ne0$.

If $q=\alpha$, the shifted pair has image $\alpha(1,v[1])$, in the image of $H$ for every $v$. To describe these pairs, note that $[p,\alpha]\in L$ with $\alpha\ne0$. The coefficients $\theta_d$ in $\vc g$ are precisely the coefficients of $[p,\alpha]$ in the fixed basis of $L$, so determining $[p,\alpha]$ also determines every added term.

For $L=\Delta$, this forces $p=\alpha$. Only the basis generator $[1,1]$ occurs, so the shifted pair is $\alpha[1+t,1+b_{\Delta,0}]$. No $M_j$ contributes such a pair, since its second entry is divisible by $x_j$ and cannot be a nonzero constant. For $L=E$, write $[p,\alpha]=[a+h,2a+h]$ with $h\in\Span_{\Q}\{f_1,\ldots,f_m\}$. The equation $2a+h=\alpha$ makes $h$ constant; the standing assumption on the equation span gives $h=0$ and $a=\alpha/2$. Only the basis generator $[1,2]$ occurs, so the shifted pair is $\alpha[\tfrac12+t,1+\tfrac12b_{E,0}]$. These are precisely the other two lines in~\eqref{eq:three-lines}.
\end{proof}

\subsection{Completing separation using a rational root}
Lemma~\ref{lem:root-free-separation} separates every outside pair off the three lines in~\eqref{eq:three-lines} without using a root. At a rational root $\vc c$, let both forms evaluate polynomials in $\Q[\vc x]_{\le2}$ at $\vc c$, set every $u[b_{L,d}]=v[b_{L,d}]=v[t]=0$, and put $u[t]=\tau>0$. Every such extension separates the nonzero multiples of $\vc e$, so we choose $\tau$ to separate the nonzero pairs on the two shifted lines as well. Their representatives have images $(1+\tau,1)$ and $(\tfrac12+\tau,1)$, which avoid the horizontal image of $K$ and the image $\Q(\tau,1)$ of $H$. Each $\widehat L$ image excludes at most one value of $\tau$ for each representative, so one positive rational value separates every nonzero multiple of either representative while preserving separation of the nonzero multiples of $\vc e$.

\begin{samepage}
\begin{lemma}\label{lem:pointwise-completion}\prooflink{lem:pointwise-completion}
If the quadratic system has a rational solution, every pair in $Y$ can be separated from $\mathcal D$.
\end{lemma}
\end{samepage}

When the system has a rational solution, Lemma~\ref{lem:pointwise-completion} supplies a separator for every outside pair; Lemma~\ref{lem:finite-subcollection} selects a finite list, completing the first reduction.

\section{From finite separation to stratifiability}\label{sec:geometry}
Let $D_1,\ldots,D_k\le\Q^r\times\Q^r$ be an instance of finite separation, with $u\times v$ as in Section~\ref{sec:transfer}. We construct a semilinear set that is stratifiable exactly when this instance admits a finite separating family.

To construct the semilinear output, we represent an integer coordinate as a difference of two nonnegative coordinates. We keep all $x$ positions before all $y$ positions so that the two forms act on consecutive groups of letter blocks. Within each entry, the positive and negative blocks are indexed in the same order, so matching pairs for distinct indices alternate; this is the arrangement used to extract the forms from a stratification (Lemma~\ref{lem:period-forms}). We therefore order the coordinates as
\[
 (x_1^+,\ldots,x_r^+,\ x_1^-,\ldots,x_r^-,\
   y_1^+,\ldots,y_r^+,\ y_1^-,\ldots,y_r^-).
\]
Define the signed-coordinate map $\mathsf{A}\colon\Q^{4r}\to\Q^{2r}$ and the sets $Y,S$ by
\begin{equation}\label{eq:observed}
 \begin{aligned}
 Y&=\Q^{2r}\setminus\bigcup_{i=1}^k D_i,\\
 \mathsf{A}(\vc{x}^+,\vc{x}^-,\vc{y}^+,\vc{y}^-)&=(\vc{x}^+-\vc{x}^-,\vc{y}^+-\vc{y}^-),
 \qquad S=\mathsf{A}^{-1}(Y)\cap\N^{4r}.
 \end{aligned}
\end{equation}
The set $Y$ consists of the rational pairs outside the deleted union. The set $S$ contains every nonnegative integer representation of each integer pair in $Y$. These representations are not unique: for example, $(x_i^+,x_i^-)=(1,0)$ and $(2,1)$ both represent $x_i=1$. More generally, increasing matching positive and negative coordinates by the same nonnegative integer leaves their difference unchanged. Keeping every representation supplies the grids of cancellation directions used to extract separators from a stratification.

\begin{lemma}\label{lem:semilinearity}\prooflink{lem:semilinearity}
The set $S$ in \eqref{eq:observed} is semilinear. From rational spanning vectors for each $D_i$, one can compute a finite semilinear presentation of $S$.
\end{lemma}

To obtain a stratification from a finite separating family, we first consider the signed representations of pairs separated by one fixed separator $u\times v$:
\begin{equation}\label{eq:separated-region}
 S_{u,v}=\left\{\vc a\in\N^{4r}:
 (u\times v)(\mathsf{A}\vc a)\notin\bigcup_{i=1}^k (u\times v)(D_i)\right\}.
\end{equation}
Every $S_{u,v}$ is contained in $S$.

\begin{lemma}\label{lem:separated-region}\prooflink{lem:separated-region}
The set $S_{u,v}$ is stratifiable. Consequently, a finite family of separators whose sets $S_{u,v}$ cover $S$ gives a stratification of $S$.
\end{lemma}
\begin{proof}[Proof sketch]
After clearing the coefficient denominators of $u,v$, membership in $S_{u,v}$ is a Presburger condition on two integer-weighted sums over consecutive coordinate blocks. Lemma~\ref{lem:weighted} and Theorem~\ref{thm:GS} give stratifiability.
\end{proof}

\begin{lemma}\label{lem:finite-necessity}\prooflink{lem:finite-necessity}
If $S$ is stratifiable, there is a finite list of separators such that every rational pair in $Y$ is separated from $D_1,\ldots,D_k$ by some member of the list.
\end{lemma}
\begin{proof}[Proof sketch]
Adding the same amount to a positive coordinate and its matching negative coordinate preserves the represented pair. Thus the signed representations of an outside integer pair contain a grid of independent cancellation directions. In a finite stratification, some component meeting this grid has nonnegative period combinations with image zero under $\mathsf{A}$ that span all those directions.

For each such component, let $L$ be the rational span of the images under $\mathsf{A}$ of periods occurring in these combinations. These period images can be added and subtracted within the component's image. The support restriction gives nonzero forms $u,v$ with $\ker u\times\ker v\subseteq L$. Fix one such separator for every component with the spanning property. A failure to separate a rational pair $\vc w\in Y$ from $D_i$ implies $\vc w\in D_i+L$. Choose one expression as a point of $D_i$ plus a rational combination of those period images for each candidate component and each $D_i$ whenever possible. Clear all these coefficients and the coordinates of $\vc w$ by one positive integer multiple before selecting a component representing that multiple. Subtracting the resulting integer combination would put a point of $D_i$ in the component's image under $\mathsf{A}$, contradicting its containment in $Y$.
\end{proof}

\begin{theorem}[Second reduction]\label{thm:separation-transfer}
There is a computable many-one reduction from finite separation by Cartesian products of rational linear forms to stratifiability. On input $D_1,\ldots,D_k$, it outputs a finite semilinear presentation of the set $S$ in \eqref{eq:observed}.
\end{theorem}
\begin{proof}
Lemma~\ref{lem:semilinearity} computes the required presentation from the subspace generators. If a finite separating family exists, the sets $S_{u,v}$ indexed by its members cover $S$: every represented integer pair lies outside the subspaces and is separated by a member of the family. Lemma~\ref{lem:separated-region} makes their finite union stratifiable. Conversely, if $S$ is stratifiable, Lemma~\ref{lem:finite-necessity} supplies a finite family separating every rational pair outside the subspaces.
\end{proof}

\section{Proof of the Main Theorem}\label{sec:main-proof}
\begin{proof}[Proof of Theorem~\ref{thm:main}]
  Given the input polynomial $f\in\Z[x_1,\ldots,x_n]$, the standard reduction to
  quadratic systems \cite[proof of Theorem~7.5, pp.~262--263]{Davis} computes a
  finite system $\mathcal E$ of rational polynomial equations of degree at most
  two that has a rational solution if and only if $f$ has a rational zero.

Theorem~\ref{thm:quadratic-transfer} computes a family of rational subspaces admitting a finite separating family exactly when $\mathcal E$ has a rational solution. Theorem~\ref{thm:separation-transfer} then computes a finite semilinear presentation of a set $S$ that is stratifiable exactly when that separating family exists. Thus $S$ is stratifiable exactly when $f$ has a rational zero.
\end{proof}

\begin{corollary}[Conditional consequence]\label{cor:undecidable}
If the undecidability claim in \cite[Theorem~1.1]{H10} is correct, general stratifiability of semilinear sets is undecidable. Equivalently, there is no algorithm that, from a semilinear set $S$, decides whether the letter-bounded language $L(S)$ is context-free.
\end{corollary}
\begin{proof}
Under the stated hypothesis, rational polynomial solvability is undecidable. Apply Theorem~\ref{thm:main}, and then Theorem~\ref{thm:GS} for the language formulation.
\end{proof}

\section{Conclusion}\label{sec:conclusion}
We have given a computable many-one reduction from rational polynomial
solvability to stratifiability of semilinear sets and, equivalently, to
context-freeness of letter-bounded languages specified by semilinear exponent
sets. If the undecidability claim in \cite{H10} is correct, both decision
problems are undecidable.

Let $H_{10}(\N)$ denote the problem of deciding whether a polynomial with
integer coefficients has a zero in nonnegative integers. Stratifiability is
computably enumerable: enumerate candidate stratifications and check equality
with the input set in Presburger arithmetic. The
Davis--Putnam--Robinson--Matiyasevich theorem \cite{Davis} therefore gives a
computable many-one reduction from stratifiability to $H_{10}(\N)$. Can one
prove an unconditional computable many-one reduction in the opposite
direction, from $H_{10}(\N)$ to stratifiability? This would establish
undecidability independently of \cite{H10}.
Adapting the proof of Lemma~\ref{lem:pointed-quadratic} to $H_{10}(\N)$
would require a mechanism forcing the recovered coordinates to be
nonnegative integers. A fixed system of polynomial equations with rational
unknowns expressing membership in $\N$ would provide such a mechanism;
finding one is equivalent to the longstanding open problem of existentially
defining $\Z$ in $\Q$ \cite{Poonen}.

Our result places no fixed bound on the dimension of the semilinear set.
For which fixed dimensions is stratifiability decidable? Every semilinear
subset of $\N^2$ is stratifiable, since both restrictions on periods hold
automatically. Dimension three is the first nontrivial case:
$\{(n,n,n):n\in\N\}$ is not stratifiable, since any nonzero period of a
linear set contained in it has all three coordinates nonzero. Is
stratifiability decidable for arbitrary semilinear subsets of $\N^3$?

If \cite[Corollary~7.1]{H10} is correct, stratifiability and rational
polynomial solvability are Turing equivalent: each can be decided using
oracle queries to the other problem. This does not settle the stronger
many-one converse, which the cited draft does not establish. Can one
effectively associate a single integer polynomial $f_S$ with every
semilinear presentation of a set $S$, so that $f_S$ has a rational zero
exactly when $S$ is stratifiable?

\parhead{AI disclosure.}
The ideas and writing in this paper were developed in conjunction with ChatGPT-6 Astra. The author takes full responsibility for the correctness and presentation of the work.

\clearpage
\appendix
\section{Proofs for the first reduction}\label{sec:first-proofs}
We first prove the finite-subcollection lemma for an arbitrary finite family of rational subspaces. For the subsequent proofs, the subspaces, bases, and coefficient notation are those of Section~\ref{sec:transfer}: $\mathcal L=\{\Delta,M_1,\ldots,M_n,E\}$, $\mathcal D=\{\widehat{L}:L\in\mathcal L\}\cup\{H,K\}$ is its enlarged family, and $Y$ consists of the pairs outside $\bigcup_{D\in\mathcal D}D$.

\subsection{Selecting finitely many separators}\label{app:finite-subcollection}
\begin{proof}[Proof of Lemma~\ref{lem:finite-subcollection}]\proofback{lem:finite-subcollection}
For each separator $T=u\times v$ in the given collection, failure to separate from $D_i$ means membership in the rational subspace
\[
 B_{T,i}=T^{-1}(T(D_i)).
\]
Choose rational linear forms $\ell_{T,i,j}$ whose simultaneous vanishing defines $B_{T,i}$, and define the polynomial
\[
 h_T(\vc z)=\prod_{i=1}^k\left(\sum_j\ell_{T,i,j}(\vc z)^2\right)
 \qquad(\vc z\in\Q^{2r}).
\]
A sum of rational squares is zero exactly when every summand is zero. Consequently, $h_T(\vc z)\ne0$ exactly when $T$ separates $\vc z$ from every $D_i$. If $B_{T,i}$ is the whole space, use no equations for it: the empty sum is zero, as required. For $k=0$, the empty product is $1$.

The polynomials $h_T$ have degree at most $2k$ in the pair coordinates, so their span is finite-dimensional. Choose a basis of this span from the polynomials $h_T$ themselves, and retain their separators. If all the retained polynomials vanish at a pair, every $h_T$ vanishes there; conversely, a nonzero retained value gives separation by a member of the given collection. These finitely many separators therefore separate exactly the same pairs. When the span is zero, the empty subcollection suffices.
\end{proof}

\subsection{Recovering a root before enlargement}\label{app:original-recovery}
\begin{proof}[Proof of Lemma~\ref{lem:pointed-quadratic}]\proofback{lem:pointed-quadratic}
Suppose $u\times v$ separates $\vc e$. Its image is $(0,v[1])$, so $v[1]\ne0$ because every subspace image contains zero. The image of $\Delta$ is proper and contains $(u[1],v[1])$; it is therefore a line avoiding $(0,v[1])$. Hence $u[1]\ne0$, and every $(u[p],v[p])$ is proportional to $(u[1],v[1])$. Dividing $u$ and $v$ by their respective values at $1$ gives a common form $\lambda$ with $\lambda[1]=1$. This invertible rescaling of the two scalar coordinates preserves separation.

The proper image of $M_j$ contains $(1,\lambda[x_j])$, so
\[
 \lambda[x_jp]=\lambda[x_j]\lambda[p]
 \qquad(p\in\Q[\vc x]_{\le1}).
\]
Put $\vc c=(\lambda[x_1],\ldots,\lambda[x_n])$. Taking $p=x_k$ shows that $\lambda$ agrees with evaluation at $\vc c$ on every quadratic monomial, hence on every polynomial of degree at most two. The image of $E$ contains $(1,2)$ and $(f_i(\vc c),f_i(\vc c))$; their determinant is $-f_i(\vc c)$, so properness forces $f_i(\vc c)=0$ for every $i$.

Conversely, evaluation at a rational solution gives a nonzero form $\lambda$ with $\lambda[1]=1$. For $u=v=\lambda$, the images of $\Delta,M_j,E$ lie in lines of slopes $1,c_j,2$, respectively, all avoiding the image $(0,1)$ of $\vc e$.
\end{proof}

\subsection{Why \texorpdfstring{$[x,1]$}{[x,1]} cannot be separated before enlargement}\label{app:pointed-example}
\begin{proof}[Verification of Example~\ref{ex:pointed-insufficient}]\proofback{ex:pointed-insufficient}
For $f=x^2-1$, consider the three subspaces
\[
 \Delta=\Span\{[1,1],[x,x],[x^2,x^2]\},\qquad
 M=\Span\{[1,x],[x,x^2]\},\qquad
 E=\Span\{[1,2],[f,f]\}.
\]
The separator $\lambda\times\lambda$, where $\lambda[p]=p(1)$, separates $\vc e=[0,1]$ from this family, but the outside pair $[x,1]$ cannot be separated.

Indeed, $[x,1]$ is not in $\Delta$, its second entry is not $x$ times its first, and the first entry of every pair in $E$ has zero $x$ coefficient. Thus it lies outside all three spaces.

Suppose nonzero forms $u,v$ separated $[x,1]$. The image of $\Delta$ cannot be the whole plane, so $u,v$ are proportional. Both are nonzero; rescaling one of them therefore gives a common form $\lambda$ and preserves separation. Put
\[
 a=\lambda[1],\qquad b=\lambda[x],\qquad c=\lambda[x^2].
\]
The image of $M$ is generated by $(a,b)$ and $(b,c)$, and the image of $E$ by $(a,2a)$ and $(c-a,c-a)$. Neither image can be the whole plane. Their determinants consequently give
\[
 ac=b^2,\qquad a(c-a)=0.
\]
If $a=0$, then $b=0$, so the image $(b,a)$ of $[x,1]$ is zero and belongs to every image. If $a\ne0$, then $c=a$ and $b=\pm a$. When $b=a$, $(b,a)$ belongs to the diagonal image of $\Delta$. When $b=-a$, it is the negative of $(a,b)$ and belongs to the image of $M$. Every case contradicts separation.
\end{proof}

\subsection{Recovery in the enlarged space}\label{app:enlarged-recovery}
\begin{proof}[Proof of Lemma~\ref{lem:recover-pointed}]\proofback{lem:recover-pointed}
A pair in $\widehat{L}$ with all added second coordinates zero has every generator coefficient zero, so it cannot equal $\vc e$. In $H$, a nonzero constant term in the second entry comes with a nonzero coefficient of $t$ in the first; in $K$, the constant term in the second entry is zero. Thus $\vc e$ lies outside every member of $\mathcal D$.

Suppose $u\times v$ separates $\vc e$. Its image is $(0,v[1])$, with $v[1]\ne0$. If $u[t]=0$, the generator $[t,1]\in H$ would have the same image, so $u[t]\ne0$. The image of $K$ therefore contains the nonzero horizontal vector $(u[t],0)$. Any nonzero $v[b_{L,d}]$ would make this image the whole plane, preventing separation. Consequently every $v[b_{L,d}]=0$, and
\[
 (u\times v)[p_{L,d},q_{L,d}+b_{L,d}]
   =\bigl(u[p_{L,d}],v[q_{L,d}]\bigr).
\]
Thus $(u\times v)(\widehat{L})=(u\times v)(L)$ for every $L\in\mathcal L$. The restriction of $v$ to $\Q[\vc x]_{\le2}$ is nonzero since $v[1]\ne0$. If $u[1]=0$, the generator $[1,1+b_{\Delta,0}]$ would have the same image as $\vc e$; hence the restriction of $u$ is also nonzero. These restrictions separate $\vc e$ from $\mathcal L$.
\end{proof}

\subsection{Why equal forms do not suffice}\label{app:equal-forms}
\begin{proof}[Verification of Example~\ref{ex:equal-forms}]\proofback{ex:equal-forms}
A pair in $\widehat{L}$ with zero-valued added second coordinates must be zero, while the first entry of every pair in $H$ or $K$ has zero coefficient at $1$. Thus $[1,1]$ lies outside every member of $\mathcal D$. If $u[b_{\Delta,0}]=0$, the generator $[1,1+b_{\Delta,0}]$ has the same image under $u\times u$ as $[1,1]$. If $u[b_{\Delta,0}]\ne0$, the image of $K$ contains the nonzero diagonal vector $(u[b_{\Delta,0}],u[b_{\Delta,0}])$, and hence the image $(u[1],u[1])$ of $[1,1]$. Neither case permits separation.
\end{proof}

\subsection{Separation by variable image slopes}\label{app:slope-separation}
We prove the preliminary classification~\eqref{eq:coarse-residuals} used in Lemma~\ref{lem:root-free-separation}, including its restrictions on $\alpha,\beta,q$.

\begin{proof}\proofback{coarse-classification}
Under the standing assumption $1\notin\Span_{\Q}\{f_1,\ldots,f_m\}$, each coordinate projection of every $L\in\mathcal L$ is injective; this will let us recover generator coefficients from either entry of a pair. An element of $E$ can be written as $[a+h,2a+h]$, where $a\in\Q$ and $h\in\Span_{\Q}\{f_1,\ldots,f_m\}$. If either entry is zero, then $h=-a$ or $h=-2a$ is constant. The only constant in the equation span is zero, so $h=0$ and $a=0$. Each coordinate projection of $E$ is therefore injective. The same property is immediate for $\Delta$ and for the first projection of $M_j$; the second projection of $M_j$ is injective because multiplication by $x_j$ is injective. Consequently, for each $L$, the coefficient vectors $[p_{L,d}]$ are linearly independent, as are the vectors $[q_{L,d}]$.

\emph{Pairs whose second entry has a nonzero coefficient of $t$.}
First take $u$ to extract the coefficient of $1$ and $v$ to extract the coefficient of $t$. Every pair in every member of $\mathcal D$ has zero coefficient of $t$ in its second entry, so these forms separate every pair whose second entry has a nonzero coefficient of $t$. Every pair whose second entry has zero coefficient of $t$ can be written as
\begin{equation}\label{eq:pair-normal-form}
 \vc w=\left[p+\alpha t+\sum_{L,d}\eta_{L,d}b_{L,d},\quad
 q+\sum_{L,d}\theta_{L,d}b_{L,d}\right].
\end{equation}
\emph{Assigning slopes to the modified subspaces.}
To put the image of $K$ into the diagonal and the image of each $\widehat{L}$ into a line of slope $s_L$, choose the values of $u,v$ on the monomial basis of $\Q[\vc x]_{\le2}$ and the slopes $s_L$ as independent parameters, subject to $v[1]=1$, and extend~$u,v$~by
\begin{equation}\label{eq:generic-form-values}
 u[t]=v[t]=0,\qquad
 u[b_{L,d}]=v[b_{L,d}]=s_Lu[p_{L,d}]-v[q_{L,d}].
\end{equation}
Indeed, the modified generator has image $(u[p_{L,d}],s_Lu[p_{L,d}])$. In $K$, the generator $[t,0]$ maps to zero and every $[b_{L,d},b_{L,d}]$ maps to a diagonal vector. The image of $H$ is vertical. The value $v[t]$ does not affect any of these images or the image of $\vc w$, since every second entry under consideration has zero coefficient of $t$.

To compute the image of $\vc w$ in \eqref{eq:pair-normal-form}, put
\[
 p_L^\eta=\sum_d\eta_{L,d}p_{L,d},\qquad
 p_L^\theta=\sum_d\theta_{L,d}p_{L,d},\qquad
 q^\eta=\sum_{L,d}\eta_{L,d}q_{L,d},\qquad
 q^\theta=\sum_{L,d}\theta_{L,d}q_{L,d}.
\]
For fixed $\vc w$, denote its image by $(\rho,\sigma)$. These coordinates are polynomials in the values of $u,v$ on the monomial basis of $\Q[\vc x]_{\le2}$ and in the slopes $s_L$:
\[
 \rho=u[p]+\sum_Ls_Lu[p_L^\eta]-v[q^\eta],\qquad
 \sigma=v[q-q^\theta]+\sum_Ls_Lu[p_L^\theta].
\]
To keep $(\rho,\sigma)$ off the vertical axis, the diagonal, and the lines of slopes $s_J$, it suffices to require $\rho\ne0$, $\sigma-\rho\ne0$, and $\sigma-s_J\rho\ne0$ for every $J\in\mathcal L$. We identify the pairs for which one of these polynomials is zero for every parameter choice. Write $\widehat{J}+\Q[t,0]$ for the pairs $\vc g+\alpha[t,0]$ with $\vc g\in \widehat{J}$ and $\alpha\in\Q$. In the following identities, $\equiv0$ means that a polynomial is identically zero on the parameter space $v[1]=1$:

\begin{align}
 \rho\equiv0&\iff p=0\text{ and }\eta_{L,d}=0\text{ for every }L,d,\label{eq:generic-U}\\
 \sigma-\rho\equiv0&\iff\vc w\in K,\label{eq:generic-K}\\
 \sigma-s_J\rho\equiv0&\iff\vc w\in \widehat{J}+\Q[t,0]
                              \quad(J\in\mathcal L).\label{eq:generic-G}
\end{align}
For \eqref{eq:generic-U}, the coefficients of the independent $s_L$ give $p_L^\eta=0$, hence every $\eta_{L,d}=0$ by independence of the $p_{L,d}$; then $\rho=u[p]$ forces $p=0$. For \eqref{eq:generic-K}, those coefficients give $\theta_{L,d}=\eta_{L,d}$ for every $L,d$, leaving $v[q]-u[p]$. Varying the values of $u$ on that monomial basis independently gives $p=0$, and then $q=0$. Indeed, if $v[q]=0$ for every form with $v[1]=1$, varying its values on the nonconstant monomials shows that every nonconstant coefficient of $q$ is zero; its constant coefficient must then also be zero. The conditions $p=q=0$ and $\theta_{L,d}=\eta_{L,d}$ for every $L,d$ say exactly that $\vc w\in K$.

For \eqref{eq:generic-G}, the coefficients of $s_J^2$ and $s_Js_L$ for $L\ne J$ first force all $p_L^\eta=0$, so every $\eta_{L,d}=0$. The remaining $s_L$ coefficients give $p_L^\theta=0$ for $L\ne J$ and $p_J^\theta=p$. Thus all $\theta_{L,d}$ with $L\ne J$ are zero. The term independent of the slopes is $v[q-q^\theta]$; its vanishing for every form with $v[1]=1$ gives $q=q^\theta$. These conditions say precisely that $\vc w\in \widehat{J}+\Q[t,0]$. The reverse implications in all three identities follow by substitution.

Consequently, if $\vc w\in Y$ has first entry not a scalar multiple of $[t]$ and lies in none of the spaces $\widehat{L}+\Q[t,0]$, then
\begin{equation}\label{eq:generic-product}
 \rho(\sigma-\rho)\prod_{L\in\mathcal L}(\sigma-s_L\rho)
\end{equation}
is a nonzero polynomial. A nonzero rational polynomial has a rational tuple where it is nonzero: this follows by induction on the number of variables from the fact that a nonzero univariate polynomial has only finitely many roots. Choose such a tuple for~\eqref{eq:generic-product}. Then the image of $\vc w$ lies outside every subspace image. The factor $\rho$ and the normalization $v[1]=1$ ensure that both forms are nonzero.

\emph{Making the image of $H$ zero.}
With $v[1]=1$ and $u[t]=0$, every pair whose first entry is $[\alpha t]$ maps into the vertical image of $H$, so these extensions cannot separate it. To separate such a pair, we instead make the image of $H$ zero. Choose new values of $u,v$ on the monomial basis of $\Q[\vc x]_{\le2}$, subject to $u[1]=1$ and $v[1]=0$, and new slopes $s_L$; then recompute their values on the added basis vectors by \eqref{eq:generic-form-values}. Consider a pair whose first entry is $[\alpha t]$ and whose second entry is not a scalar multiple of $[1]$. For these new extensions, the image $(\rho,\sigma)$ has $\rho=0$, and $H$ maps to zero. The polynomial~$\sigma$ is nonzero: if it vanished identically, its $s_L$ coefficients would give $u[p_L^\theta]\equiv0$ for all forms with $u[1]=1$, hence $p_L^\theta=0$ and every $\theta_{L,d}=0$. The remaining identity $v[q]\equiv0$ for forms satisfying $v[1]=0$ would make $q$ constant, a contradiction. Choose a rational solution of $\sigma\ne0$. The pair's image $(0,\sigma)$ is nonzero and avoids the finite-slope lines containing the images of all $\widehat{L},K$. Both forms are nonzero, since $u[1]=1$ and $\sigma\ne0$.

The constructions above prove separation for every pair in $Y$ except possibly those described in~\eqref{eq:coarse-residuals}. For $[\alpha t,\beta]\in Y$, membership outside $K$ gives $\beta\ne0$ and membership outside $H$ gives $\beta\ne\alpha$. If $\vc w=\vc g+\alpha[t,0]\in Y$ with $\vc g\in \widehat{L}$, then $\alpha\ne0$ because $\vc w\notin \widehat{L}$. Write the second entry of $\vc g$ as $[q+\sum_d\theta_db_{L,d}]$, with $q\in\Q[\vc x]_{\le2}$. If $q=\sum_d\theta_dq_{L,d}=0$, injectivity of the second projection of $L$ would force every $\theta_d=0$. Then $\vc g=0$ and $\vc w\in K$, a contradiction.
\end{proof}

\subsection{Completing separation using a rational root}\label{app:root-completion}
\begin{proof}[Proof of Lemma~\ref{lem:pointwise-completion}]\proofback{lem:pointwise-completion}
Fix a rational root $\vc c$. For a positive rational parameter $\tau$, define
\begin{equation}\label{eq:root-extension}
 \begin{aligned}
 u[h]&=v[h]=h(\vc c) &&(h\in\Q[\vc x]_{\le2}),\\
 u[t]&=\tau,\quad v[t]=0,\quad
 u[b_{J,d}]=v[b_{J,d}]=0 &&(\text{every }J,d).
 \end{aligned}
\end{equation}
Both forms are nonzero at $[1]$. The $\widehat{J}$ images lie in lines of slopes $1,c_1,\ldots,c_n,2$, while the image of $K$ is horizontal and $H$ maps to $\Q(\tau,1)$. Thus every nonzero multiple of $\vc e$ has a nonzero vertical image and is separated.

For the other two lines in~\eqref{eq:three-lines}, the representatives have images
\[
 (1+\tau,1),\qquad(\tfrac12+\tau,1).
\]
Both avoid the image of $K$, since their second coordinate is nonzero. Membership in the image $\Q(\tau,1)$ of $H$ would force $1=0$ or $\tfrac12=0$, respectively. For a line of slope $s$ containing a $\widehat J$ image, membership requires $1=s(1+\tau)$ or $1=s(\tfrac12+\tau)$. If $s=0$, neither is possible; otherwise these equations exclude the two values
\[
 \tau=\frac1s-1,\qquad \tau=\frac1s-\frac12.
\]
Choose one positive rational $\tau$ outside the finite set excluded by all the slopes $1,c_1,\ldots,c_n,2$. This separator separates both representatives and hence every nonzero rational multiple of either. It also preserves separation of every nonzero multiple of $\vc e$, proved for every positive $\tau$. Lemma~\ref{lem:root-free-separation} supplies separators for the pairs in $Y$ off these three lines, proving the claim.
\end{proof}

\section{Proofs for the second reduction}\label{sec:geometry-proof}
The subspaces $D_i$, map $\mathsf{A}$, sets $Y,S$, and separators $u\times v$ are as in Section~\ref{sec:geometry}. We first prove effective semilinearity of $S$, then extract finitely many separators from a stratification. Finally, we prove the weighted-sums lemma and apply it to the regions $S_{u,v}$.

\subsection{Computing a semilinear presentation}\label{app:semilinearity}
\begin{proof}[Proof of Lemma~\ref{lem:semilinearity}]\proofback{lem:semilinearity}
Gaussian elimination and denominator clearing compute integer matrices $\mathsf{M}_i$ with
$D_i=\{\vc z\in\Q^{2r}:\mathsf{M}_i\vc z=\vc0\}$. Since $\mathsf{A}$ has integer entries, so does $\mathsf{C}_i=\mathsf{M}_i\mathsf{A}$. If $\mathsf{C}_i$ has $s_i$ rows, then for $\vc n\in\N^{4r}$,
\[
 \vc n\in S\quad\Longleftrightarrow\quad
 \bigwedge_{i=1}^k\bigvee_{j=1}^{s_i}(\mathsf{C}_i\vc n)_j\ne0.
\]
An empty disjunction is false and an empty conjunction is true. This is a Presburger formula, so the effective conversion recalled in Section~\ref{sec:prelim} produces a finite semilinear presentation.
\end{proof}

\subsection{Period combinations with image zero}
Put $d=2r$ and $K_{\mathsf{A}}=\ker \mathsf{A}=\{\vc z\in\Q^{2d}:\mathsf{A}\vc z=\vc 0\}$. For nonzero rational forms $u,v$, define
\begin{equation}\label{eq:C}
 V(u,v)=\ker u\times\ker v,\qquad
 Y_{u,v}=\Q^d\setminus\bigcup_{i=1}^k\bigl(D_i+V(u,v)\bigr).
\end{equation}
Since $\ker (u\times v)=V(u,v)$, we have
\[
 D_i+V(u,v)=(u\times v)^{-1}\bigl((u\times v)(D_i)\bigr).
\]
Thus $Y_{u,v}$ is precisely the set of pairs separated by $u\times v$, and $Y_{u,v}\subseteq Y$. In the notation of \eqref{eq:separated-region},
\[
 S_{u,v}=\mathsf{A}^{-1}(Y_{u,v})\cap\N^{4r}.
\]

To extract separators from a stratification of $S$, we first need a component whose nonnegative period combinations with image zero span $K_{\mathsf{A}}$. For each outside integer pair, the following argument finds such a component containing one of its signed representations.

For $\vc{z}\in\Q^d$, the \emph{fiber} of $\mathsf{A}$ over $\vc{z}$ is $\mathsf{A}^{-1}(\vc{z})=\{\vc{n}\in\Q^{2d}:\mathsf{A}\vc{n}=\vc{z}\}$. To describe the part of a linear set in this fiber, we use the coordinatewise order on $\N^\ell$: $\vc{b}\le\vc{c}$ means $b_i\le c_i$ for every $i$.

\begin{lemma}[Dickson's lemma]\label{lem:dickson}
Every subset of $\N^\ell$ has finitely many minimal elements for the coordinatewise order, and each of its elements is above a minimal element.
\end{lemma}
\begin{proof}
Every infinite sequence in $\N$ has an infinite nondecreasing subsequence: either one value occurs infinitely often, or each bounded set of values occurs only finitely often and an increasing subsequence can be chosen. Apply this successively to the $\ell$ coordinates. Every infinite sequence in $\N^\ell$ then has two terms in coordinatewise nondecreasing order, so an infinite sequence of distinct minimal elements is impossible. For any $\vc{n}$ in the subset, the set of elements below $\vc{n}$ is finite and nonempty; a minimal member is also minimal in the original subset. The case $\ell=0$ is immediate.
\end{proof}

For a finite period set $P=\{\vc{t}_1,\ldots,\vc{t}_\ell\}\subseteq\N^{2d}$, let $\mathsf{P}$ be the matrix with these columns. Define
\[
 Z_P=\{\vc{n}\in\N^\ell:\mathsf{A}\mathsf{P}\vc{n}=\vc{0}\},\qquad
 K_P=\Span\{\mathsf{P}\vc{n}:\vc{n}\in Z_P\}\subseteq K_{\mathsf{A}}.
\]
Thus $K_P$ is spanned by the nonnegative period combinations that leave the represented pair unchanged. The next proof uses the minimal-element argument for linear-equation solution sets in Ginsburg and Spanier's proof of Theorem~1.3~\cite{GS66}.

\begin{lemma}\label{lem:fibers}
If $S=\bigcup_{j=1}^t(\vc{a}_j+P_j^*)$, then for every $\vc{z}\in Y\cap\Z^d$ some component meeting $\mathsf{A}^{-1}(\vc{z})$ satisfies $K_{P_j}=K_{\mathsf{A}}$.
\end{lemma}
\begin{proof}
First fix one component $\vc{a}+P^*$ and put
\[
 E_{\vc{b}}=\{\vc{n}\in\N^\ell:\mathsf{A}\mathsf{P}\vc{n}=\vc{b}\},\qquad \vc{b}=\vc{z}-\mathsf{A}\vc{a}.
\]
By Lemma~\ref{lem:dickson}, $E_{\vc{b}}$ has finitely many minimal elements $\vc{c}$. Every $\vc{n}\in E_{\vc{b}}$ is above one of them, and $\vc{n}-\vc{c}\in Z_P$; conversely, $\vc{c}+Z_P\subseteq E_{\vc{b}}$. Hence
\[
 E_{\vc{b}}=\bigcup_{\vc{c}\text{ minimal in }E_{\vc{b}}}(\vc{c}+Z_P).
\]
It follows that $(\vc{a}+P^*)\cap \mathsf{A}^{-1}(\vc{z})$ lies in finitely many translates $\vc{a}+\mathsf{P}\vc{c}+K_P$.

Choose $\vc{n}_{\vc{z}}\in\N^{2d}$ with $\mathsf{A}\vc{n}_{\vc{z}}=\vc{z}$. For each of the $d$ coordinates of $\Q^d$, let $\vc{k}_i$ have a $1$ in its positive and negative source positions and zeros elsewhere. The vectors $\vc{k}_1,\ldots,\vc{k}_d$ are a basis of $K_{\mathsf{A}}$, and
\[
 \vc{n}_{\vc{z}}+\sum_{i=1}^d\N\vc{k}_i\subseteq S\cap \mathsf{A}^{-1}(\vc{z}).
\]
If every component meeting this fiber had $K_{P_j}\ne K_{\mathsf{A}}$, the preceding description would cover this grid by finitely many translates of proper subspaces of $K_{\mathsf{A}}$. In the coordinates given by $\vc{k}_1,\ldots,\vc{k}_d$, each translate satisfies an equation $c_1t_1+\cdots+c_dt_d=b$ in which some $c_i$ is nonzero. It contains at most $(N+1)^{d-1}$ points of $\{0,\ldots,N\}^d$: fix all but a coordinate with nonzero coefficient, and the remaining coordinate has at most one possible value. Finitely many such translates cannot cover all $(N+1)^d$ points for arbitrarily large $N$.
\end{proof}

For a period set $P$, let $P_0$ consist of the periods $\vc{t}_i$ for which some $\vc{n}\in Z_P$ has $n_i>0$. These are the periods that occur in a nonnegative combination with image zero. Put
\[
 \mathcal M=\mathsf{A}(P^*),\qquad \Gamma=\mathsf{A}(P_0^*),\qquad
 L=\Span \Gamma,\qquad R=\Span P_0.
\]
Here $\mathcal M,\Gamma\subseteq\Z^d$, $L\le\Q^d$, and $R\le\Q^{2d}$. Both $\mathcal M$ and $\Gamma$ contain zero and are closed under addition.

\begin{lemma}\label{lem:reversible}
The set $\Gamma$ is closed under negation and $\mathsf{A}(R)=L$. If $K_P=K_{\mathsf{A}}$, then $K_{\mathsf{A}}\subseteq R$.
\end{lemma}
\begin{proof}
For $\vc{t}_i\in P_0$, choose $\vc{n}\in Z_P$ with $n_i>0$. The relation $\sum_h n_h\mathsf{A}\vc{t}_h=\vc{0}$ gives
\[
 -\mathsf{A}\vc{t}_i=(n_i-1)\mathsf{A}\vc{t}_i+\sum_{h\ne i}n_h\mathsf{A}\vc{t}_h.
\]
Every period with a positive coefficient on the right belongs to $P_0$. Thus the negative of each generator of $\Gamma$ belongs to $\Gamma$, which proves closure under negation. Also $\mathsf{A}(R)=\Span \mathsf{A}(P_0)=\nobreak L$. In every $\mathsf{P}\vc{n}$ with $\vc{n}\in Z_P$, only periods in $P_0$ have positive coefficients, so $K_P\subseteq R$.
\end{proof}
Since $\Gamma\subseteq \mathcal M$, adding or subtracting an element of $\Gamma$ preserves the image $\mathsf{A}\vc{a}+\mathcal M$ of the linear component. These moves will turn a failure of separation into a point of $\mathsf{A}\vc{a}+\mathcal M$ lying in a deleted subspace.

\subsection{Two forms from the order of period supports}
The \emph{support} of a period $\vc t$ is the set $\{i:t_i\ne0\}$ of positions of its nonzero coordinates. We next use the prohibition on alternating period supports to find $u,v$ with $V(u,v)\subseteq L$. The argument follows the structure of stratified period sets studied by Brough \cite[Lemmas~3.6--3.7 and the proof of Proposition~3.8]{Brough}. We first record the graph property we need.

\begin{lemma}\label{lem:ordered-graph}
Let a graph have linearly ordered vertices and no edges $\{a,c\},\{b,d\}$ with $a<b<c<d$. Distinct connected components cannot contain $a,c$ and $b,d$, respectively, with $a<b<c<d$.
\end{lemma}
\begin{proof}
If $\{i,k\}$ is an edge with $i<k$, any other component meeting the open interval $(i,k)$ lies entirely inside that interval. Otherwise a path leaving the interval has an edge with one endpoint inside and one outside. Its endpoints are distinct from $i,k$, so the two edges have alternating endpoints.

Suppose distinct components $C,D$ contain $a,c$ and $b,d$ with $a<b<c<d$. A path in $C$ from $a$ to $c$ has an edge $\{i,k\}$ with $i<b<k$, so $D\subseteq(i,k)$. A path in $D$ from $b$ to $d$ has an edge $\{j,\ell\}$ with $j<c<\ell$, so likewise $C\subseteq(j,\ell)$. But $i<j<\ell<k$, contradicting $i\in C$.
\end{proof}

\begin{lemma}\label{lem:period-forms}
If a subspace $R\le\Q^{2d}$ is spanned by a stratified period set and contains $K_{\mathsf{A}}$, there are nonzero rational forms $u$ on $\Q^r$ and $v$ on $\Q^r$ with $V(u,v)\subseteq \mathsf{A}(R)$.
\end{lemma}
\begin{proof}
Let $\mathcal H$ be the space of coefficient vectors of rational linear forms on $\Q^{2d}$ that vanish on $R$. Since these forms vanish on $K_{\mathsf{A}}$, every vector in $\mathcal H$ can be written as
\begin{equation}\label{eq:opposite-coefficients}
 \vc{w}=(\vc{\alpha},-\vc{\alpha},\vc{\beta},-\vc{\beta}),\qquad
 \vc{\alpha}\in\Q^r,\quad\vc{\beta}\in\Q^r.
\end{equation}
Construct a graph on the source coordinate positions, with an edge for each two-element support of a spanning period. For a period $\vc{t}$ supported on $\{i,j\}$, vanishing means
\[
 t_iw_i+t_jw_j=0,\qquad t_i,t_j\ne0.
\]
Along an edge one coefficient determines the other. On a connected component, all coefficients are therefore determined by any one of them. The space of restrictions of $\vc{w}\in\mathcal H$ to a component therefore has dimension at most one. This holds for isolated vertices as well.

Restricting $\vc{w}$ to one graph component and setting its other coefficients to zero preserves membership in $\mathcal H$, because each spanning period has support in a single component. Call the two positions $x_i^+,x_i^-$, or $y_i^+,y_i^-$, \emph{signed partners}. If a coefficient can be nonzero in $\mathcal H$, its position and its signed partner must belong to the same component. Otherwise, choose $\vc{w}\in\mathcal H$ nonzero at that position and restrict $\vc{w}$ to that graph component; the resulting vector violates~\eqref{eq:opposite-coefficients}.

All positions in the two $\vc{x}$ blocks whose coefficients can be nonzero belong to one component. Indeed, two indices $i<j$ in distinct components would put $x_i^+,x_i^-$ in one and $x_j^+,x_j^-$ in the other. These four positions alternate in the source order, contrary to Lemma~\ref{lem:ordered-graph}. The same holds for the $\vc{y}$ blocks. Thus the possible vectors $\vc{\alpha}$ span a space of dimension at most one, as do the possible vectors $\vc{\beta}$.

If the space of possible $\vc{\alpha}$ is nonzero, choose a form $u$ whose coefficient vector spans it; otherwise choose any nonzero form on $\Q^r$. Choose $v$ in the same way from the possible $\vc{\beta}$. Every rational form on $\Q^d$ that vanishes on $\mathsf{A}(R)$, when composed with $\mathsf{A}$, has a coefficient vector in $\mathcal H$. The form on $\Q^d$ can therefore be written as
\[
 (\vc{x},\vc{y})\longmapsto c\,u(\vc{x})+e\,v(\vc{y})\qquad(c,e\in\Q),
\]
and vanishes on $V(u,v)$. This implies $V(u,v)\subseteq \mathsf{A}(R)$: if a vector were outside $\mathsf{A}(R)$, extending a basis of $\mathsf{A}(R)$ would give a form vanishing there but not at that vector.
\end{proof}

\subsection{A stratification supplies separating forms}
For a component of a stratification with $K_P=K_{\mathsf{A}}$, the preceding lemmas give forms with $V(u,v)\subseteq L$. If these forms fail to separate a rational pair $\vc w\in Y$ from $D_i$, then $\vc w\in D_i+L$. To turn the $L$ contribution into an element of $\Gamma$ that can be subtracted within the component's image, we clear denominators simultaneously for all candidate components before applying Lemma~\ref{lem:fibers} to select one representing the scaled pair.

\begin{lemma}\label{lem:necessity}
If $S$ is stratifiable, finitely many separators separate every pair in $Y$.
\end{lemma}
\begin{proof}\proofback{lem:finite-necessity}
Fix a stratification $S=\bigcup_{j=1}^t(\vc{a}_j+P_j^*)$, and let
\[
 J=\{j:K_{P_j}=K_{\mathsf{A}}\}.
\]
For each $j\in J$, use the notation $\mathcal M_j,\Gamma_j,L_j,R_j$ defined above. Lemma~\ref{lem:reversible} gives $K_{\mathsf{A}}\subseteq R_j$ and $\mathsf{A}(R_j)=L_j$. Since the periods spanning $R_j$ are a subset of $P_j$, they are stratified. Lemma~\ref{lem:period-forms} supplies nonzero rational forms $u_j,v_j$ with
\begin{equation}\label{eq:kernel-inclusion}
 V(u_j,v_j)\subseteq L_j.
\end{equation}
We show that these finitely many separators separate all of $Y$.

Fix $\vc{w}\in Y$. For every index combination $(j,i)$ with $j\in J$ and $\vc{w}\in L_j+D_i$, choose an expression
\[
 \vc{w}=\vc{d}_{ji}+\sum_h c_{jih}\vc{h}_{jih},\qquad
 \vc{d}_{ji}\in D_i,\quad\vc{h}_{jih}\in \Gamma_j,\quad c_{jih}\in\Q.
\]
Such expressions exist because $L_j=\Span \Gamma_j$. The fiber lemma will select a component only after we fix an integer multiple of $\vc w$, so we clear denominators simultaneously for every candidate component. Choose one positive integer $N$ clearing the denominators of the coordinates of $\vc{w}$ and all the coefficients $c_{jih}$. As $\Gamma_j$ is closed under addition and negation, while $D_i$ is a subspace, we obtain
\begin{equation}\label{eq:clear-pair}
 N\vc{w}\in \Gamma_j+D_i\quad\text{whenever }\vc{w}\in L_j+D_i.
\end{equation}
Also $N\vc{w}\in Y\cap\Z^d$, since multiplication by a nonzero scalar preserves membership in each $D_i$ and its complement.

Lemma~\ref{lem:fibers} gives $j\in J$ with $N\vc{w}\in \mathsf{A}\vc{a}_j+\mathcal M_j$. If $\vc{w}\in L_j+D_i$, equation~\eqref{eq:clear-pair} gives $N\vc{w}=\vc{h}+\vc{d}'$ for $\vc{h}\in \Gamma_j$ and $\vc{d}'\in D_i$. Since $-\vc{h}\in \Gamma_j\subseteq \mathcal M_j$,
\[
 \vc{d}'=N\vc{w}-\vc{h}\in \mathsf{A}\vc{a}_j+\mathcal M_j\subseteq Y,
\]
contradicting $\vc{d}'\in D_i$. Therefore $\vc{w}\notin L_j+D_i$ for every $i$. By~\eqref{eq:kernel-inclusion}, $\vc{w}\in Y_{u_j,v_j}$, which proves the claim.
\end{proof}

\subsection{Separating forms supply a stratification}\label{app:weighted-regions}
For a source vector $\vc{n}=(\vc{x}^+,\vc{x}^-,\vc{y}^+,\vc{y}^-)$ and a fixed separator $u\times v$, membership of $\mathsf{A}\vc{n}$ in $Y_{u,v}$ depends on the two sums $u(\vc{x}^+-\vc{x}^-)$ and $v(\vc{y}^+-\vc{y}^-)$. After clearing the coefficient denominators of $u,v$, these are integer-weighted sums over consecutive coordinate blocks. The following lemma proves that a Presburger condition on two such sums gives a context-free letter-bounded language.

\begin{samepage}
\begin{lemma}\label{lem:weighted}
Let $a_1,\ldots,a_s,b_1,\ldots,b_t$ be distinct letters, let $h_i,k_j\in\Z$, and let $T\subseteq\Z^2$ be Presburger-definable. The language of words
\[
 a_1^{n_1}\cdots a_s^{n_s}b_1^{m_1}\cdots b_t^{m_t}
 \quad\text{such that}\quad
 \left(\sum_i h_i n_i,\sum_j k_jm_j\right)\in T
\]
is context-free.
\end{lemma}
\end{samepage}
\begin{proof}
Intersect $T$ with each of the four closed coordinate quadrants. Changing coordinate signs identifies each intersection with a semilinear subset of $\N^2$. For one of its linear components $\vc{c}+\sum_\ell\N(\alpha_\ell,\beta_\ell)$, choose fresh terminal letters $\mathtt{a},\mathtt{b}$, distinct from every $a_i,b_j$. The grammar
\[
 X\longrightarrow \mathtt{a}^{\alpha_\ell}X\mathtt{b}^{\beta_\ell}\quad(\text{each }\ell),
 \qquad X\longrightarrow \mathtt{a}^{c_1}\mathtt{b}^{c_2}
\]
generates exactly its two-block language. Thus each such linear component has a context-free two-block language.

On the alphabet $\{+,-\}$ define the \emph{balance} of a word to be the number of plus signs minus the number of minus signs. We will replace a count $n\in\N$ by words of balance $\epsilon n$, where $\epsilon\in\{1,-1\}$ is the sign assigned to that coordinate by the chosen quadrant. Let $\mathcal B_j$ be the language of words of balance $j$. A grammar for $\mathcal B_0$ is
\[
 Z\longrightarrow\varepsilon\mid +Z-Z\mid -Z+Z.
\]
The languages $\mathcal B_1=\mathcal B_0+\mathcal B_0$ and $\mathcal B_{-1}=\mathcal B_0-\mathcal B_0$ are also context-free, where juxtaposition denotes concatenation and the middle sign is a literal letter. For $\epsilon\in\{1,-1\}$ and $n\in\N$,
\begin{equation}\label{eq:balance}
 \mathcal B_0(\mathcal B_\epsilon)^n=\mathcal B_{\epsilon n}.
\end{equation}
For the reverse inclusion, cut a word when its balance first reaches $\epsilon$, then $2\epsilon$, and so on up to $n\epsilon$. Each intervening piece has balance $\epsilon$; append the remaining balanced suffix to the~last~piece. The case $n=0$ is immediate. The same first-passage argument gives the descriptions of $\mathcal B_1$ and $\mathcal B_{-1}$ above.

Insert new marker letters $\#_1,\#_2$ before the two count blocks: add the start production $X_0\to \#_1X$ and replace the terminal production by $X\to \mathtt{a}^{c_1}\#_2\mathtt{b}^{c_2}$. Use two disjoint copies of the signed alphabet. Substitute $\mathcal B_0$ for each marker, $\mathcal B_{\epsilon_1}$ for $\mathtt{a}$, and $\mathcal B_{\epsilon_2}$ for $\mathtt{b}$, where $\epsilon_1,\epsilon_2$ are the chosen signs. By \eqref{eq:balance}, the resulting language consists exactly of concatenations of a word over the first signed alphabet and a word over the second, whose balance pair is $(\epsilon_1n,\epsilon_2m)$ for some $(n,m)$ in the selected linear component. Take the union over components and coordinate-sign regions. Finally, for nonzero weights, map each $a_i$ to $|h_i|$ copies of the sign of $h_i$ in the first alphabet, and each $b_j$ to $|k_j|$ copies of the sign of $k_j$ in the second. Map each letter of weight zero to the empty word. Take the inverse image under this homomorphism and intersect it with the regular language $a_1^*\cdots a_s^*b_1^*\cdots b_t^*$. This is the required language. These constructions use closure of context-free languages under concatenation, finite union, substitution of context-free languages for letters, inverse homomorphism, and intersection with a regular language \cite[Sections~3.4, 6.4, and~6.5]{Harrison}.
\end{proof}

\begin{proof}[Proof of Lemma~\ref{lem:separated-region}]\proofback{lem:separated-region}
Multiplying each form by a positive integer to clear its coefficient denominators preserves separation: it applies the same invertible coordinate map to the image of each pair and each subspace. Membership in $S_{u,v}$ is therefore a Presburger condition on the two integer-weighted sums
\[
 u(\vc{x}^+-\vc{x}^-),\qquad v(\vc{y}^+-\vc{y}^-).
\]
The pair of sums must lie outside $\bigcup_i(u\times v)(D_i)$, a union of rational subspaces of $\Q^2$ defined by linear equations. All $\vc x$ positions precede all $\vc y$ positions. Lemma~\ref{lem:weighted} therefore shows that the associated letter-bounded language is context-free. Theorem~\ref{thm:GS} therefore makes $S_{u,v}$ stratifiable. Finite unions preserve stratifiability.
\end{proof}

\end{document}